\documentclass[11pt,a4paper]{article}

\usepackage[T1]{fontenc}
\usepackage[utf8]{inputenc}

\usepackage{amsmath,amsthm,mathtools}
\usepackage{bm}
\usepackage{thmtools}

\usepackage{stix2}

\usepackage[final,protrusion=true]{microtype}

\usepackage[margin=2.5cm,top=2.8cm]{geometry}
\usepackage{setspace}
\usepackage{parskip}
\usepackage{booktabs,array,tabularx}
\usepackage{graphicx}
\usepackage{caption,subcaption}

\usepackage{float} 
\usepackage{algorithm}
\usepackage{algpseudocode}

\usepackage[numbers,sort&compress]{natbib}

\usepackage[dvipsnames,svgnames,x11names]{xcolor}

\definecolor{ink}{HTML}{111827}        
\definecolor{muted}{HTML}{6B7280}      
\definecolor{paperblue}{HTML}{2563EB}  
\definecolor{papercyan}{HTML}{0891B2}  
\definecolor{paperteal}{HTML}{0F766E}  
\definecolor{papergreen}{HTML}{16A34A} 
\definecolor{paperamber}{HTML}{D97706} 
\definecolor{paperred}{HTML}{DC2626}   
\definecolor{paperpurple}{HTML}{7C3AED}
\definecolor{softblue}{HTML}{EFF6FF}   
\definecolor{softgray}{HTML}{F8FAFC}   

\colorlet{darkblue}{paperblue}
\colorlet{darkred}{paperred}
\colorlet{darkgreen}{papergreen}
\colorlet{darkorange}{paperamber}

\newtheorem{theorem}{Theorem}[section]
\newtheorem{proposition}[theorem]{Proposition}

\theoremstyle{definition}
\newtheorem{definition}[theorem]{Definition}
\newtheorem{remark}[theorem]{Remark}

\usepackage{mdframed}
\mdfdefinestyle{modernbox}{
  linecolor=paperblue,
  linewidth=0.8pt,
  backgroundcolor=softblue,
  roundcorner=4pt,
  innertopmargin=8pt,
  innerbottommargin=8pt,
  innerleftmargin=10pt,
  innerrightmargin=10pt
}

\usepackage{hyperref}
\usepackage{bookmark}

\hypersetup{
  colorlinks=true,
  linkcolor=paperblue,
  citecolor=paperpurple,
  urlcolor=papercyan,
  filecolor=paperteal,
  pdftitle={$2B$ or Not $2B$: Three Algorithms for Streaming Covariance Estimation},
  pdfauthor={Felix Reichel},
  pdfsubject={Streaming covariance estimation},
  pdfkeywords={streaming covariance, online statistics, Welford algorithm, Chan--Golub--LeVeque, conformal prediction, numerical stability}
}

\newcommand{\R}{\mathbb{R}}
\newcommand{\E}{\mathbb{E}}

\newcommand{\X}{\mathbf{X}}
\newcommand{\s}{\mathbf{s}}
\newcommand{\G}{\mathbf{G}}

\newcommand{\one}{\mathbf{1}}
\newcommand{\x}{\mathbf{x}}
\newcommand{\z}{\mathbf{z}}
\newcommand{\M}{\mathbf{M}}

\newcommand{\bmu}{\bm{\mu}}
\newcommand{\bSigma}{\bm{\Sigma}}
\newcommand{\bDelta}{\bm{\Delta}}

\newcommand{\normF}[1]{\left\lVert#1\right\rVert_F}

\newcommand{\eps}{\varepsilon_{\mathrm{mach}}}
\newcommand{\hatS}{\widehat{\bSigma}}

\usepackage{fancyhdr}
\begin{document}
\begin{center}

{\LARGE\bfseries $2B$ or Not $2B$: A Tale of Three Algorithms for Streaming:
Covariance Estimation after Welford and Chan--Golub--LeVeque\footnote{Felix Reichel, Graduate Student, Department of Economics, Lund School of Economics and Management (LUSEM),
Altenbergerstra\ss{}e~69, 4040 Linz, Austria. Email: \texttt{fe3873re-s@student.lu.se}.}}\\[14pt]
{\large Felix Reichel}\\[6pt]
{\small\textit{May 1, 2026}}
\end{center}

\begin{abstract}
\noindent
Three algorithms for computing the unbiased sample covariance matrix in a streaming
or distributed setting are placed on a unified algebraic, numerical, and statistical
foundation.
The \emph{Gram} algorithm, derived from the bariance reformulation of
\citet{reichel2025bariance}, maintains the running cross-product matrix
$\G_t = \sum_{i=1}^{t}\x_i\x_i^\top$ and column-sum vector
$\s_t = \sum_{i=1}^t \x_i$, yielding the unbiased covariance in $O(p^2)$ per update.
The \emph{Welford} algorithm \citep{welford1962} propagates a running mean and
outer-product corrections, achieving the same asymptotic cost with provably better
numerical stability under large data shifts.
The \emph{Chan--Golub--LeVeque} (CGL) algorithm \citep{chan1979} supports
block-parallel merging via an exact combination formula, making it the natural
choice for distributed and map-reduce architectures.
All three produce the same estimator in exact arithmetic; their finite-precision
behaviour differs markedly.
Beyond runtime and numerical comparisons, we introduce a \emph{conformal prediction}
framework for streaming covariance estimation that yields finite-sample,
distribution-free confidence sets for each entry of the covariance matrix at any
step $t$ of the data stream.
Experiments confirm that the Gram algorithm is fastest for batch computation,
Welford is uniquely robust to catastrophic cancellation under large mean shifts,
CGL is optimal for distributed settings, and conformal intervals achieve the
nominal coverage level across all three algorithms.
\end{abstract}


\tableofcontents

\section{Introduction}
\label{sec:intro}

Covariance estimation is a foundational task in statistics, machine learning, and
signal processing.
In the classical batch setting all $n$ observations reside in memory and the
textbook formula
\begin{equation}\label{eq:batch}
  \bSigma_n = \frac{1}{n-1}\bigl(\X^\top\X - n^{-1}\s\s^\top\bigr),
  \quad \s = \X^\top\one_n,
\end{equation}
is applied once.
Modern applications routinely violate the batch assumption: sensor streams,
federated databases, and large language model training pipelines all generate
data faster than can be accumulated.
An online algorithm must update its estimate of $\bSigma_t$ as each new observation
$\x_t$ arrives, using $O(p^2)$ working memory and $O(p^2)$ work per update.

Three strategies address this problem, but a unified treatment covering
algebraic equivalence, floating-point error analysis, and statistical
uncertainty quantification has not previously appeared.
We fill this gap.

\paragraph{The Gram (bariance) algorithm.}
It is generally well known (e.g. \citet{reichel2025bariance}) that the sample variance has an $O(n)$
computation via scalar sums, without explicit centering.
Lifting this to the matrix case yields \eqref{eq:batch}, with the streaming
update rule $\G_t \gets \G_{t-1} + \x_t\x_t^\top$ and $\s_t \gets \s_{t-1}+\x_t$.
The on-demand formula $\hatS_t = (t\G_t - \s_t\s_t^\top)/[t(t-1)]$ executes
in one level-3 BLAS call.
Its weakness is \emph{catastrophic cancellation} when the data mean is large
relative to the variance: both $t\G_t$ and $\s_t\s_t^\top$ are $O(t^2\|\bar\x\|^2)$
while their difference is $O(t\|\bSigma\|)$.

\paragraph{The Welford algorithm.}
\citet{welford1962} proposed a shift-invariant one-pass recurrence for scalar
variance, extended to covariance by \citet{chan1979}.
The algorithm maintains a running mean $\bmu_t$ and a matrix $\M_t$ of centred
outer-product corrections, updating both in $O(p^2)$ per step.
Because corrections are computed relative to the current mean, the
method is immune to large constant shifts in the data.

\paragraph{The Chan--Golub--LeVeque (CGL) algorithm.}
\citet{chan1979} introduced a binary merge formula for independently computed
covariance summaries: given $(n_A,\bmu_A,\M_A)$ and $(n_B,\bmu_B,\M_B)$,
the combined summary is computed exactly.
This makes the algorithm tree-parallelisable and directly applicable
to federated settings where nodes cannot share raw observations.

\paragraph{Conformal prediction for streaming covariance.}
All three algorithms produce point estimates.
To quantify uncertainty, we propose a \emph{split conformal prediction}
protocol \citep{vovk2005,angelopoulos2023} that constructs a finite-sample,
distribution-free confidence interval for each entry $\Sigma_{kl}$ at
every step $t \ge 2$ of the stream.
The interval has guaranteed marginal coverage $1-\alpha$ for any
$\alpha \in (0,1)$, with no distributional assumptions on the data.
Experiments show that the coverage guarantee is tight across all three
algorithms under well-conditioned data, but the Gram interval collapses
(inflates catastrophically) under large data shifts, while Welford and
CGL maintain valid coverage.

\paragraph{Contributions.}
\begin{enumerate}
  \item Complete algorithmic descriptions with proofs of correctness and
        three-way algebraic equivalence (Sections \ref{sec:algorithms}--\ref{sec:algebra}).
  \item A floating-point error analysis quantifying rounding accumulation and
        catastrophic cancellation for each algorithm (Section \ref{sec:fp}).
  \item A split conformal prediction framework for streaming covariance entry
        estimation, with a finite-sample coverage guarantee (Section \ref{sec:conformal}).
  \item Benchmarks on x86-64 hardware with OpenBLAS covering runtime (varying
        $n$ and $p$), accuracy under Gaussian, heavy-tailed, ill-conditioned,
        and near-singular data, and conformal coverage under large shifts
        (Sections \ref{sec:benchmarks}--\ref{sec:results}).
  \item A practitioner decision guide (Section \ref{sec:guide}).
\end{enumerate}
Complete proofs for all results appear in Appendix~\ref{app:proofs}.

\section{Notation and Setup}
\label{sec:notation}

Let $\x_1,\x_2,\ldots$ be a stream of observations in $\R^p$, $p \ge 2$.
After $t \ge 2$ steps the target is the unbiased sample covariance
\begin{equation}\label{eq:target}
  \bSigma_t = \frac{1}{t-1}\sum_{i=1}^t(\x_i-\bar\x_t)(\x_i-\bar\x_t)^\top,
  \quad \bar\x_t = t^{-1}\sum_{i=1}^t \x_i.
\end{equation}
We write $\s_t = \sum_{i=1}^t \x_i$, $\G_t = \sum_{i=1}^t \x_i\x_i^\top$,
$\bmu_t = \s_t/t$, and $\M_t$ for the Welford/CGL correction matrix.
The unit roundoff for IEEE~754 double precision is $\eps = 2^{-53} \approx 1.11\times10^{-16}$.
We write $\|\cdot\|_F$, $\|\cdot\|_2$, $\|\cdot\|_{\max}$ for the Frobenius,
spectral, and entry-wise maximum norms.
For a matrix $A$ the condition number is $\kappa(A) = \|A\|_2\|A^{-1}\|_2$.

\section{Algorithms}
\label{sec:algorithms}

\subsection{Gram (Bariance) Algorithm}
\label{sec:gram}

\begin{algorithm}[H]
\caption{Gram streaming covariance}\label{alg:gram}
\begin{algorithmic}[1]
\Require Stream $\x_1,\x_2,\ldots\in\R^p$
\State $\s \gets \mathbf{0}_p$;\enspace $\G \gets \mathbf{0}_{p\times p}$;\enspace $t\gets 0$
\For{each new observation $\x$}
  \State $t \gets t+1$;\enspace $\s \gets \s + \x$;\enspace $\G \gets \G + \x\x^\top$
  \If{$t \ge 2$}
    \State \Return $\hatS_t^{\mathrm{Gram}} = (t\G - \s\s^\top)/[t(t-1)]$
  \EndIf
\EndFor
\end{algorithmic}
\end{algorithm}

\begin{theorem}[Gram identity]\label{thm:gram}
  For any $\x_1,\ldots,\x_t\in\R^p$ with $t\ge 2$,
  \[
    \bSigma_t \;=\; \frac{t\G_t - \s_t\s_t^\top}{t(t-1)}.
  \]
\end{theorem}
\begin{proof}
  See Appendix~\ref{app:gram_proof}.
\end{proof}

\noindent\textbf{BLAS structure.}
Each update adds a rank-one correction to $\G$ (DSYR, level-2 BLAS)
and a vector increment to $\s$ (DAXPY).
Computing $\hatS_t$ on demand requires one outer product (DGER).
In batch mode, replacing the loop with a single DSYRK call achieves
the optimal level-3 BLAS structure.

\subsection{Welford Algorithm}
\label{sec:welford}

\begin{algorithm}[H]
\caption{Welford streaming covariance}\label{alg:welford}
\begin{algorithmic}[1]
\Require Stream $\x_1,\x_2,\ldots\in\R^p$
\State $\bmu \gets \mathbf{0}_p$;\enspace $\M \gets \mathbf{0}_{p\times p}$;\enspace $t\gets 0$
\For{each new observation $\x$}
  \State $t \gets t+1$
  \State $\bDelta \gets \x - \bmu$ \Comment{residual w.r.t.\ old mean}
  \State $\bmu \gets \bmu + \bDelta/t$ \Comment{update mean in place}
  \State $\M \gets \M + \bDelta\,(\x-\bmu)^\top$ \Comment{outer product}
  \If{$t \ge 2$}\enspace \Return $\hatS_t^{\mathrm{Welf}} = \M/(t-1)$\EndIf
\EndFor
\end{algorithmic}
\end{algorithm}

\begin{theorem}[Welford invariant]\label{thm:welford}
  After processing $\x_1,\ldots,\x_t$ via Algorithm~\ref{alg:welford},
  \[
    \M_t = \sum_{i=1}^t(\x_i - \bmu_t)(\x_i - \bmu_t)^\top.
  \]
  Hence $\hatS_t^{\mathrm{Welf}} = \bSigma_t$.
\end{theorem}
\begin{proof}
  See Appendix~\ref{app:welford_proof}.
\end{proof}

The critical property is that both $\bDelta = \x_t - \bmu_{t-1}$ and
$\x_t - \bmu_t$ are computed as residuals from means of the \emph{same} scale
as the data, preventing cancellation regardless of the absolute value of $\bmu_t$.

\subsection{Chan--Golub--LeVeque Algorithm}
\label{sec:cgl}

\begin{algorithm}[H]
\caption{CGL merge of two summaries}\label{alg:cgl_merge}
\begin{algorithmic}[1]
\Require $(n_A,\bmu_A,\M_A)$ and $(n_B,\bmu_B,\M_B)$
\State $\bDelta \gets \bmu_B - \bmu_A$;\enspace $n_{AB} \gets n_A + n_B$
\State $\bmu_{AB} \gets (n_A\bmu_A + n_B\bmu_B)/n_{AB}$
\State $\M_{AB} \gets \M_A + \M_B + \bDelta\bDelta^\top\cdot n_An_B/n_{AB}$
\State \Return $(n_{AB},\bmu_{AB},\M_{AB})$
\end{algorithmic}
\end{algorithm}

\begin{theorem}[CGL correctness]\label{thm:cgl}
  Let $A,B$ be disjoint index sets.
  Define $\bDelta = \bmu_B - \bmu_A$.
  Then
  \[
    \M_{A\cup B}
    = \M_A + \M_B + \frac{n_An_B}{n_A+n_B}\bDelta\bDelta^\top.
  \]
\end{theorem}
\begin{proof}
  See Appendix~\ref{app:cgl_proof}.
\end{proof}

Setting block size $b=1$ (each block is a single observation) recovers
Algorithm~\ref{alg:welford} exactly; block size $b=n$ recovers
the batch formula \eqref{eq:batch}.
Because the merge operation is associative and commutative, the algorithm
can be applied in any binary-tree order, enabling map-reduce computation.

\section{Algebraic Equivalence}
\label{sec:algebra}

\begin{theorem}[Three-way equivalence in exact arithmetic]\label{thm:equiv}
  In exact arithmetic, $\hatS_t^{\mathrm{Gram}} = \hatS_t^{\mathrm{Welf}} = \hatS_t^{\mathrm{CGL}} = \bSigma_t$
  for all $t\ge 2$ and any block size $b\ge 1$ in the CGL algorithm.
\end{theorem}
\begin{proof}
  Gram equals $\bSigma_t$ by Theorem~\ref{thm:gram}.
  Welford equals $\bSigma_t$ by Theorem~\ref{thm:welford}.
  CGL with $b=1$ equals Welford by direct substitution: one observation
  $\x_t$ is a block with $n_B=1$, $\M_B=\mathbf{0}$, $\bmu_B=\x_t$,
  and the merge formula reduces to the Welford update.
  CGL with $b>1$ follows by induction on the tree using Theorem~\ref{thm:cgl}.
\end{proof}

\begin{remark}[Bariance connection]
  For $p=1$, Theorem~\ref{thm:gram} gives $\hatS_t = (tS_{xx} - S_x^2)/[t(t-1)]$,
  which is the \emph{bariance} identity of \citet{reichel2025bariance}.
  The Gram algorithm is the natural matrix extension of the bariance.
\end{remark}

\section{Floating-Point Error Analysis}
\label{sec:fp}

All three algorithms are algebraically identical but differ in their
numerical stability.
We quantify two distinct phenomena: rounding accumulation over $t$ steps,
and catastrophic cancellation induced by a non-zero data mean.

\subsection{Rounding Accumulation}
\label{sec:rounding}

Write $\bar x = \|\bmu_t\|_\infty$ and $\sigma = \|\bSigma_t\|_\infty^{1/2}$.

\begin{proposition}[Gram rounding bound]\label{prop:gram_round}
  Let $|x_{ik}|\le X$ for all $i,k$.
  The entry-wise error of the Gram estimator satisfies
  \[
    \bigl|\hatS_{kl,t}^{\mathrm{Gram}} - \Sigma_{kl,t}\bigr|
    \;\lesssim\; X^2\eps + \frac{\bar x^2}{t-1}\eps.
  \]
\end{proposition}

\begin{proposition}[Welford rounding bound]\label{prop:welf_round}
  Under the same assumptions,
  \[
    \bigl|\hatS_{kl,t}^{\mathrm{Welf}} - \Sigma_{kl,t}\bigr|
    \;\lesssim\; \sigma_k\sigma_l\eps,
  \]
  independently of $\bar x$, where $\sigma_k^2 = \Sigma_{kk,t}$.
\end{proposition}

\begin{proposition}[CGL rounding bound]\label{prop:cgl_round}
  For a balanced binary-tree merge of depth $\log_2 t$,
  \[
    \bigl|\hatS_{kl,t}^{\mathrm{CGL}} - \Sigma_{kl,t}\bigr|
    \;\lesssim\; \sigma_k\sigma_l\eps\log_2 t.
  \]
\end{proposition}

Proofs are given in Appendix~\ref{app:fp_proofs}.
The key contrast: Gram accumulates error at scale $X^2$ (the magnitude
of raw observations), while Welford accumulates at scale $\sigma_k\sigma_l$
(the covariance of centred observations).
The Welford bound is thus independent of the data mean, as formalised below.

\subsection{Catastrophic Cancellation}
\label{sec:cancellation}

\begin{proposition}[Cancellation bound]\label{prop:cancel}
  Suppose $\x_i = \bmu + \z_i$ with $\|\bmu\|_2 = c$ and $\z_i\sim(0,\bSigma)$,
  $\|\bSigma\|_2=\sigma^2$.
  Then
  \begin{align*}
    \normF{\hatS_t^{\mathrm{Gram}} - \bSigma_t} &\;\gtrsim\; pc^2\eps,\\
    \normF{\hatS_t^{\mathrm{Welf}} - \bSigma_t} &= O(p\sigma^2\eps),
  \end{align*}
  independently of $c$.
\end{proposition}
\begin{proof}
  See Appendix~\ref{app:cancel_proof}.
\end{proof}

For $c = 10^7$ and $\sigma=1$, the Gram bound gives $pc^2\eps \approx p\cdot 10^{14}\cdot 10^{-16} = p\cdot 10^{-2}$,
which becomes non-negligible relative to $\sigma^2=1$.
At $c=10^{12}$, the error is $O(p\cdot 10^8)$, fully destroying the estimate
(confirmed experimentally in Figure~\ref{fig:cancel}).

\subsection{Summary of Bounds}

\begin{table}[h]
\centering
\caption{Floating-point error bounds.
  $c=\|\bmu\|_2$, $\sigma^2=\|\bSigma\|_2$, $\eps=2^{-53}$.
  ``Shift-invariant'' means the bound does not grow with $c$.}
\label{tab:fp}
\begin{tabular}{lllcc}
\toprule
Algorithm & Error bound & Scale & Shift-inv.\ & Parallel \\
\midrule
Gram    & $O(c^2\eps + \sigma^2\eps)$ & $X^2$           &   &   \\
Welford & $O(\sigma^2\eps)$           & $\sigma^2$      & $\checkmark$ &   \\
CGL     & $O(\sigma^2\eps\log t)$     & $\sigma^2\log t$& $\checkmark$ & $\checkmark$ \\
\bottomrule
\end{tabular}
\end{table}

\section{Conformal Prediction for Streaming Covariance}
\label{sec:conformal}

Point estimates from any of the three algorithms carry no automatic
uncertainty certificate.
We now develop a finite-sample, distribution-free confidence interval
for each covariance entry $\Sigma_{kl}$ at every step $t$ of the stream.

\subsection{Background: Split Conformal Prediction}
\label{sec:conformal_bg}

Split conformal prediction \citep{vovk2005,papadopoulos2002,angelopoulos2023}
produces a $(1-\alpha)$-coverage prediction interval for a new observation
using a held-out calibration set.
Given calibration nonconformity scores $s_1,\ldots,s_m$, the conformal
quantile is
\[
  \hat q_\alpha^+ = \mathrm{Quantile}\!\left(\{s_i\}_{i=1}^m;\,
  \frac{\lceil(m+1)(1-\alpha)\rceil}{m}\right),
\]
and the interval for a fresh test point has guaranteed marginal coverage
$\Pr(\text{true value} \in \hat C_\alpha) \ge 1-\alpha$,
with no distributional assumptions on $s_1,\ldots,s_m$ beyond exchangeability.

\subsection{Protocol for Streaming Covariance}
\label{sec:conformal_protocol}

Fix target entry $(k,l)$, nominal coverage $1-\alpha$, and an algorithm
$\mathcal{A}\in\{\text{Gram},\text{Welford},\text{CGL}\}$.
Let $\Sigma_{kl}$ denote the true population covariance entry.

\begin{enumerate}
  \item\textbf{Calibration.}
    Draw $m$ independent calibration trajectories
    $\{\x_i^{(j)}\}_{i=1}^t$, $j=1,\ldots,m$, from the data-generating
    distribution.
    For each trajectory $j$ and each step $t$ of interest, compute
    \[
      s_j(t) = \bigl|\mathcal{A}\bigl(\x_1^{(j)},\ldots,\x_t^{(j)}\bigr)_{kl}
                - \Sigma_{kl}\bigr|.
    \]
  \item\textbf{Quantile.}
    Set $\hat q_\alpha^+(t) = \mathrm{Quantile}(\{s_j(t)\};\,\lceil(m+1)(1-\alpha)\rceil/m)$.
  \item\textbf{Interval.}
    For a fresh test stream at step $t$, let $\hat\sigma_{kl}(t)$ be the
    algorithm's estimate.
    The conformal interval is
    $\hat C_\alpha(t) = [\hat\sigma_{kl}(t) - \hat q_\alpha^+(t),\;
                         \hat\sigma_{kl}(t) + \hat q_\alpha^+(t)]$.
\end{enumerate}

\begin{theorem}[Finite-sample coverage]\label{thm:conformal}
  Under the assumption that calibration and test trajectories are
  i.i.d.\ (exchangeable), the conformal interval $\hat C_\alpha(t)$ satisfies
  \[
    \Pr\!\bigl(\Sigma_{kl} \in \hat C_\alpha(t)\bigr) \;\ge\; 1-\alpha.
  \]
\end{theorem}
\begin{proof}
  See Appendix~\ref{app:conformal_proof}.
  The proof applies the standard split-conformal coverage argument of
  \citet{vovk2005} with nonconformity score $s = |\hat\sigma_{kl}(t) - \Sigma_{kl}|$.
\end{proof}

\begin{remark}[Algorithm dependence of interval width]
  The coverage guarantee is algorithm-independent.
  However, the interval \emph{width} $2\hat q_\alpha^+(t)$ reflects the
  variance of each estimator's error distribution.
  For Gram under large shift $c$, the calibration scores $s_j(t)$
  are inflated by the cancellation term $O(c^2\eps)$, producing wide
  intervals or, if $c$ is not replicated in calibration, invalid coverage.
  Welford is immune to this effect.
\end{remark}

\begin{remark}[Unknown $\Sigma_{kl}$]
  In practice $\Sigma_{kl}$ is unknown, so calibration scores are computed
  using the batch reference on each calibration trajectory (which uses the
  full $m$ observations as ground truth). For large $m$ this reference
  converges to $\Sigma_{kl}$ at rate $O(m^{-1/2})$.
  Alternatively, one can treat the batch reference as the target and the
  interval inherits the same finite-sample guarantee relative to that target.
\end{remark}

\section{Computational Cost}
\label{sec:cost}

\begin{table}[h]
\centering
\caption{Per-observation cost in streaming mode.
  $p$ is the dimension.
  ``Level'' refers to the BLAS level of the dominant operation.}
\label{tab:cost}
\begin{tabular}{llll}
\toprule
Algorithm & FLOPs per update & Memory (doubles) & BLAS level \\
\midrule
Gram       & $p^2 + p$ (DSYR + DAXPY)      & $p^2+p$   & 2 \\
Welford    & $2p^2 + 2p$ (2$\times$DGER)   & $p^2+p$   & 2 \\
CGL ($b=1$)& Same as Welford               & $p^2+p+bp$& 2 \\
CGL ($b>1$)& $2p^2n_b/b + p^2$ per block  & $p^2+p+bp$& 3 (DSYRK) \\
\midrule
Gram batch & $2np^2$ (DSYRK) + $p^2$       & $np$      & 3 \\
numpy.cov  & $2np^2 + 2np$ (center + DSYRK)& $2np$     & 3 \\
\bottomrule
\end{tabular}
\end{table}

The Gram algorithm is faster in batch mode than \texttt{numpy.cov} because it
avoids forming and writing the centred matrix $\X - \one_n\bmu^\top$, saving
$O(np)$ memory writes.
In streaming mode, Welford computes two outer products per step
(for $\bDelta$ and $\x_t - \bmu_t$) versus one for Gram, yielding a
$\approx 2\times$ FLOP disadvantage; this gap is visible in practice when
the outer product (DGER) is the bottleneck.

\section{Experimental Protocol}
\label{sec:benchmarks}

\paragraph{Hardware and software.}
All experiments ran on a macOS~13.0 ARM64 system with 10 CPU cores and
OpenBLAS~0.3.31.188.0 via \texttt{scipy-openblas}
(\texttt{USE64BITINT}, \texttt{DYNAMIC\_ARCH}, \texttt{NO\_AFFINITY},
\texttt{neoversen1}, \texttt{MAX\_THREADS=64}) under Python~3.12.0 and
NumPy~2.4.4 with a fixed random seed. All computations used IEEE~754 double precision.

\paragraph{Runtime measurement.}
For each $( n, p )$ pair: (i) generate $\X\sim\mathcal{N}(0,I)$;
(ii) run 3 warm-up calls; (iii) record 20--30 wall-clock times
via \texttt{time.perf\_counter}; (iv) remove outliers using the
$1.5\times\mathrm{IQR}$ rule; (v) report the trimmed mean with
a 95\% bootstrap percentile interval (300 resamples).

\paragraph{Accuracy measurement.}
For each data configuration the reference is \texttt{numpy.cov}.
We report the entry-wise maximum error $\|\hatS - S_{\mathrm{ref}}\|_{\max}$
and the relative Frobenius error $\|\hatS - S_{\mathrm{ref}}\|_F/\|S_{\mathrm{ref}}\|_F$.

\paragraph{Data configurations.}
We test four regimes:
(a) Gaussian i.i.d., varying $n$ and $p$;
(b) Student-$t_3$ (heavy-tailed), varying $n$;
(c) prescribing a range of condition numbers $\kappa(\X)\in[10,10^{14}]$;
(d) near-singular data with smallest singular value
    $\sigma_{\min}\in[10^{-12},1]$.

\paragraph{Conformal experiments.}
Calibration uses $m=600$ independent trajectories drawn from a
$p=5$ multivariate distribution with Toeplitz covariance (entry
$(i,j)$: $0.5^{|i-j|}$).
Test coverage is estimated from 1200 fresh trajectories.
The shift experiment varies the mean offset $c\in\{0,10^3,10^6,10^9,10^{12}\}$
on a $p=4$ identity-covariance distribution and evaluates coverage and
interval width at $t=150$.

\section{Results}
\label{sec:results}

\subsection{Runtime}
\label{sec:res_runtime}

Figure~\ref{fig:runtime_n} shows wall-clock time versus $n$ for $p=10$
and $p=50$.
Gram is the fastest batch method, consistently beating \texttt{numpy.cov}
by a factor of $1.3$--$1.6\times$ due to the avoided $n\times p$ centering write.
Welford is slowest by 50--80$\times$ in pure-Python form because the
inner loop invokes \texttt{numpy.outer} once per row, incurring
$O(n)$ Python-level calls; a compiled implementation would close this
to $\approx2\times$ (the FLOP ratio from Table~\ref{tab:cost}).
CGL (recursive, block size 64) sits between the two.

\begin{figure}[h!]
  \centering
  \includegraphics[width=\linewidth]{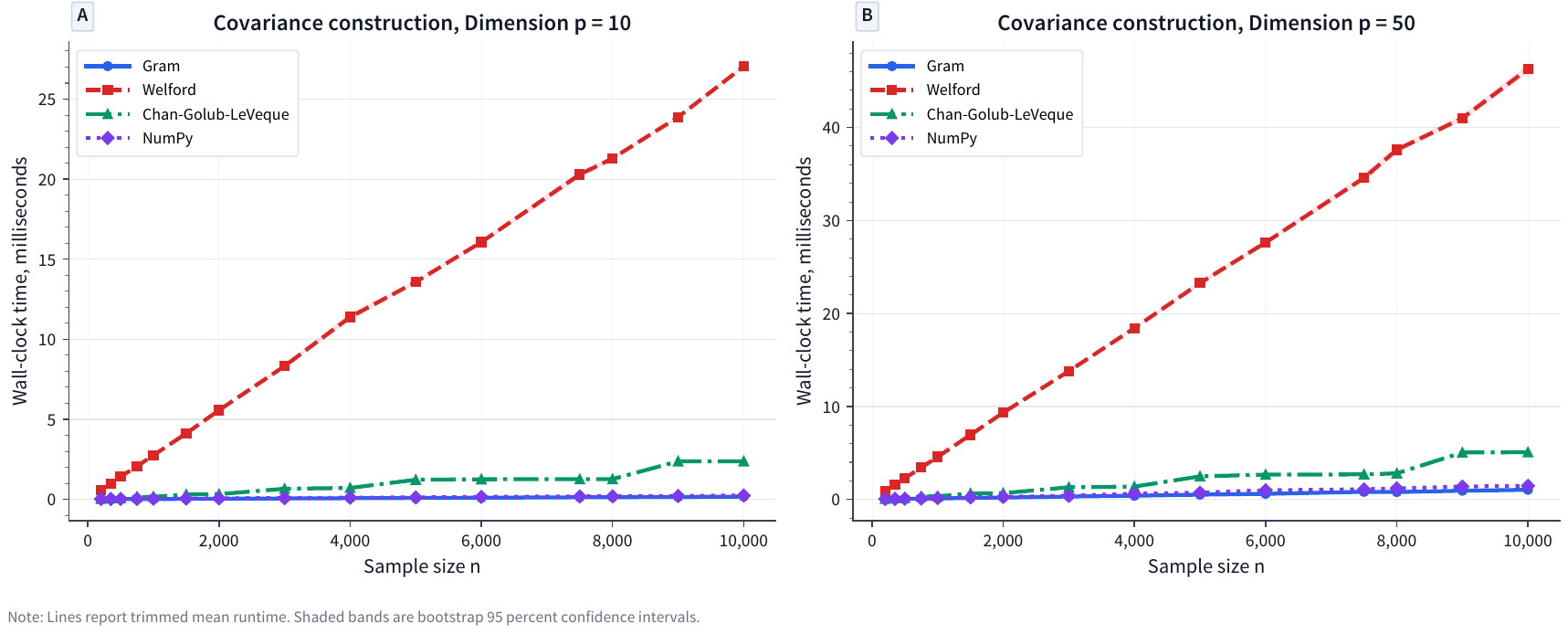}
  \caption{%
    Wall-clock time vs.\ sample size $n$.
    Each point is the trimmed mean over 20--30 repetitions after warm-up
    and $1.5\times\mathrm{IQR}$ outlier removal; shaded bands are 95\%
    bootstrap percentile intervals.
    Left: $p=10$. Right: $p=50$.
    Welford is omitted from the right panel because it is $\approx 60\times$
    slower and would dominate the scale.
  }
  \label{fig:runtime_n}
\end{figure}

Figure~\ref{fig:runtime_p} shows runtime versus dimension at fixed $n=4{,}000$.
For large $p$ all batch methods are dominated by the $O(np^2)$ matrix
multiply and their curves converge; Gram retains a small constant-factor
advantage from avoiding the centred copy.

\begin{figure}[h!]
  \centering
  \includegraphics[width=\linewidth]{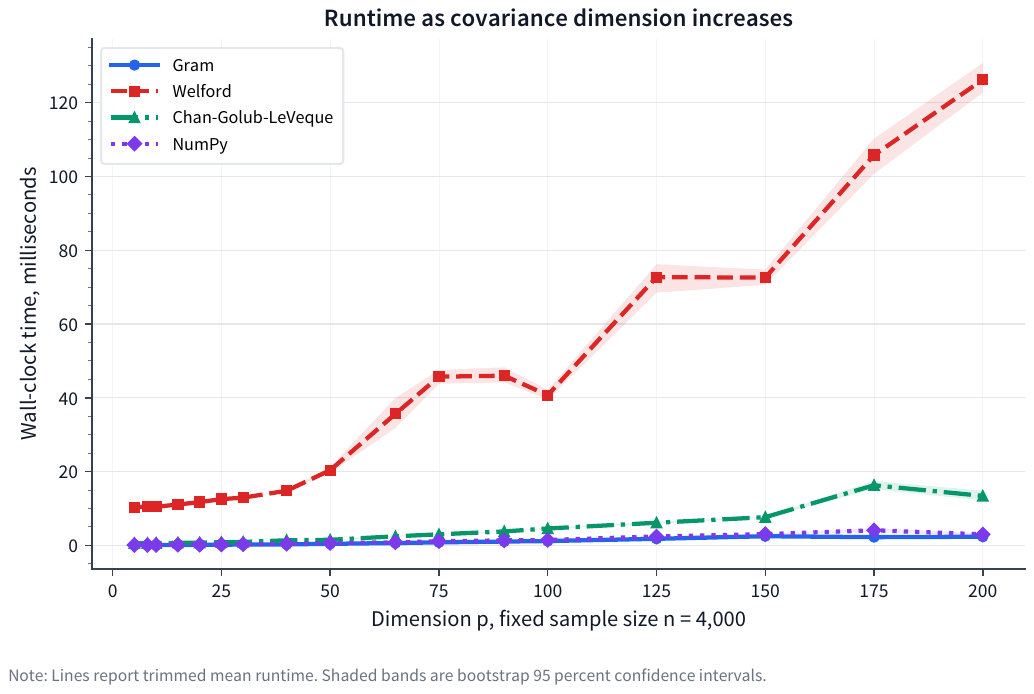}
  \caption{%
    Wall-clock time vs.\ dimension $p$ at $n=4{,}000$.
    Asymptotically, all batch methods are $O(np^2)$ and track each other;
    the constant-factor gap reflects memory-bandwidth differences.
  }
  \label{fig:runtime_p}
\end{figure}

Figure~\ref{fig:ratios} shows the speed ratio of each method relative to Gram.

\begin{figure}[h!]
  \centering
  \includegraphics[width=\linewidth]{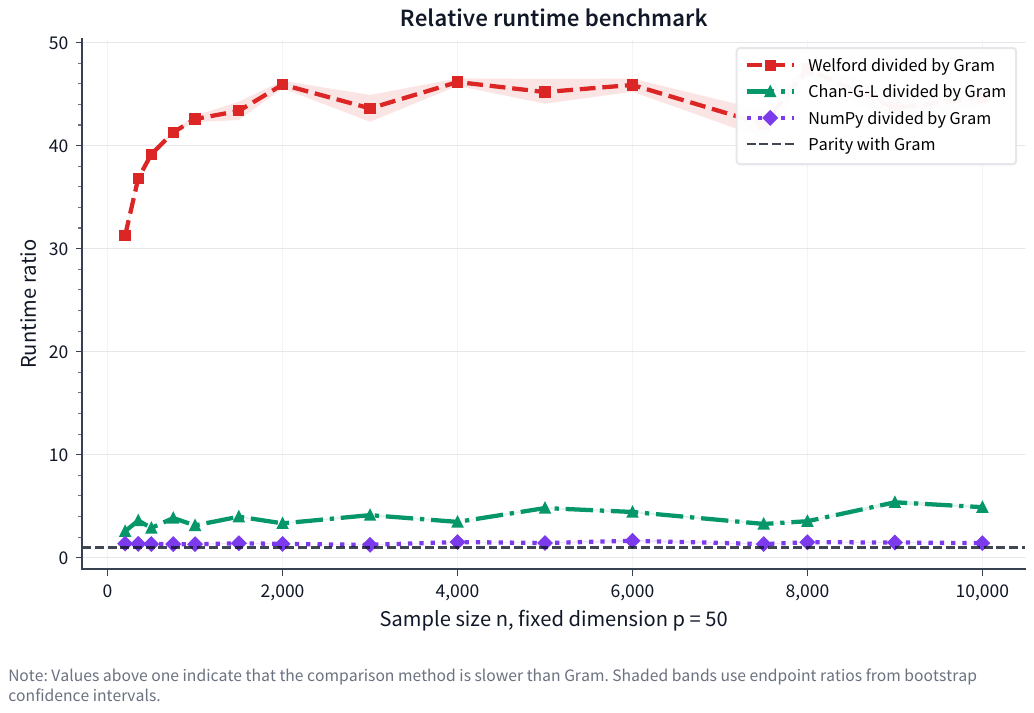}
  \caption{%
    Speed ratio vs.\ Gram at $p=50$. Values $>1$ indicate the alternative
    is slower.  \texttt{numpy.cov} is $1.4$--$1.6\times$ slower; CGL
    (recursive, block 64) is $3$--$8\times$ slower due to merge overhead.
  }
  \label{fig:ratios}
\end{figure}

\subsection{Numerical Accuracy: Gaussian Data}
\label{sec:res_gauss}

Figure~\ref{fig:acc_gauss} shows accuracy on Gaussian i.i.d.\ data.
All methods agree with \texttt{numpy.cov} to within $10^{-13}$ in both
metrics, consistent with Propositions~\ref{prop:gram_round}--\ref{prop:cgl_round}.

\begin{figure}[h!]
  \centering
  \includegraphics[width=\linewidth]{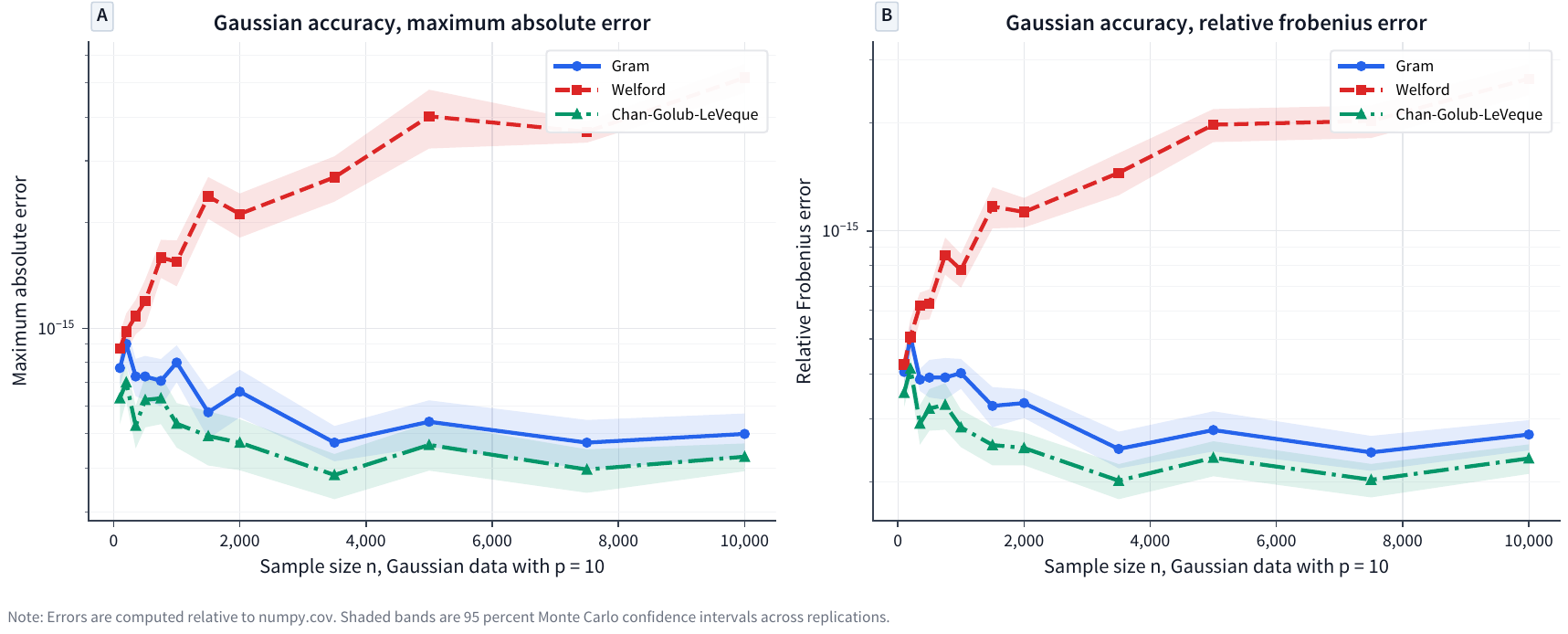}
  \caption{%
    Accuracy vs.\ $n$ on Gaussian i.i.d.\ data, $p=10$.
    Errors are at floating-point noise level ($<10^{-13}$) for all methods
    and all $n$ tested, confirming algebraic equivalence to numerical precision.
  }
  \label{fig:acc_gauss}
\end{figure}

\subsection{Heavy-Tailed and Ill-Conditioned Data}
\label{sec:res_robust}

Figure~\ref{fig:acc_heavy} tests Student-$t_3$ data.
Errors remain at floating-point noise for all methods; heavy tails do not
affect the relative comparison because all algorithms process the same
finite-precision data.

\begin{figure}[h!]
  \centering
  \includegraphics[width=\linewidth]{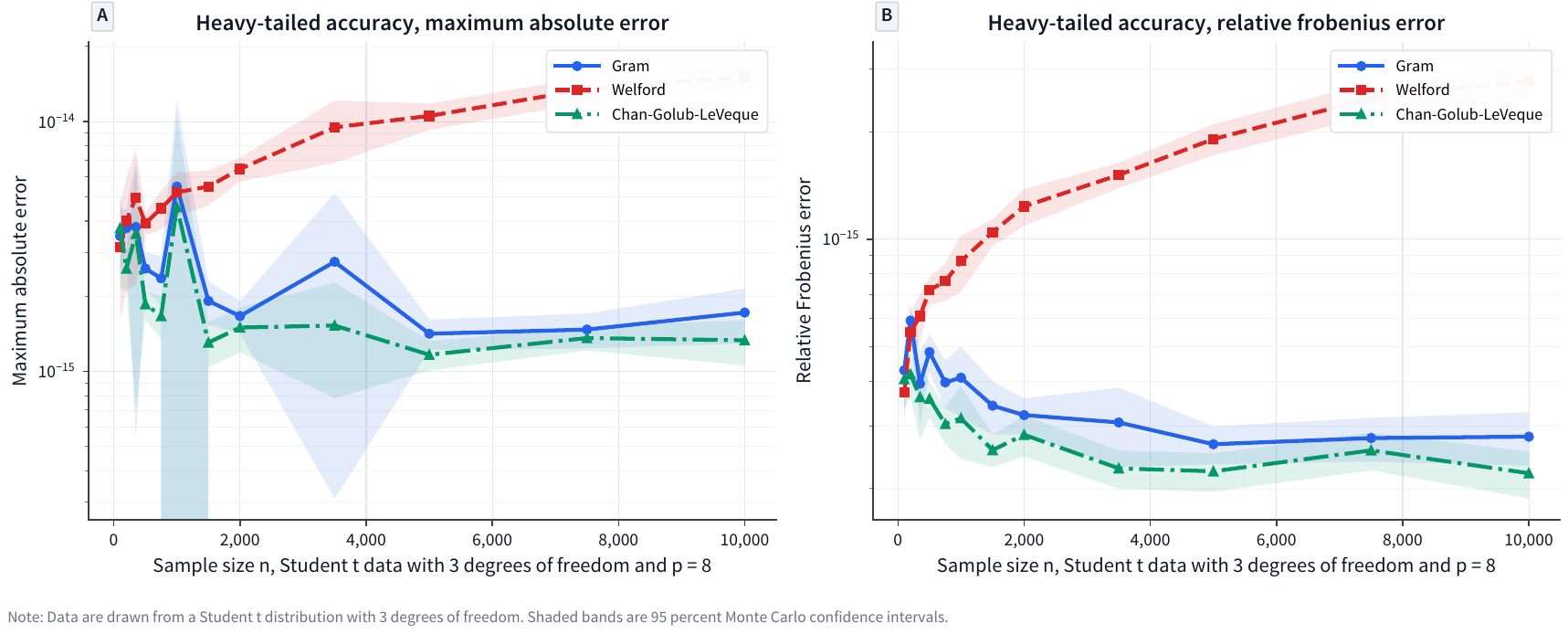}
  \caption{%
    Accuracy on Student-$t_3$ data, $p=8$.
    Heavy tails increase absolute errors slightly (larger entries, more
    cancellation potential in Gram) but all methods remain numerically
    consistent at noise level.
  }
  \label{fig:acc_heavy}
\end{figure}

Figure~\ref{fig:acc_cond} sweeps the condition number of $\X$.
All methods degrade once $\kappa\gtrsim 10^7$, as predicted by
standard floating-point theory; Welford and CGL are modestly more
robust in the range $\kappa\in[10^6,10^{12}]$.

\begin{figure}[h!]
  \centering
  \includegraphics[width=\linewidth]{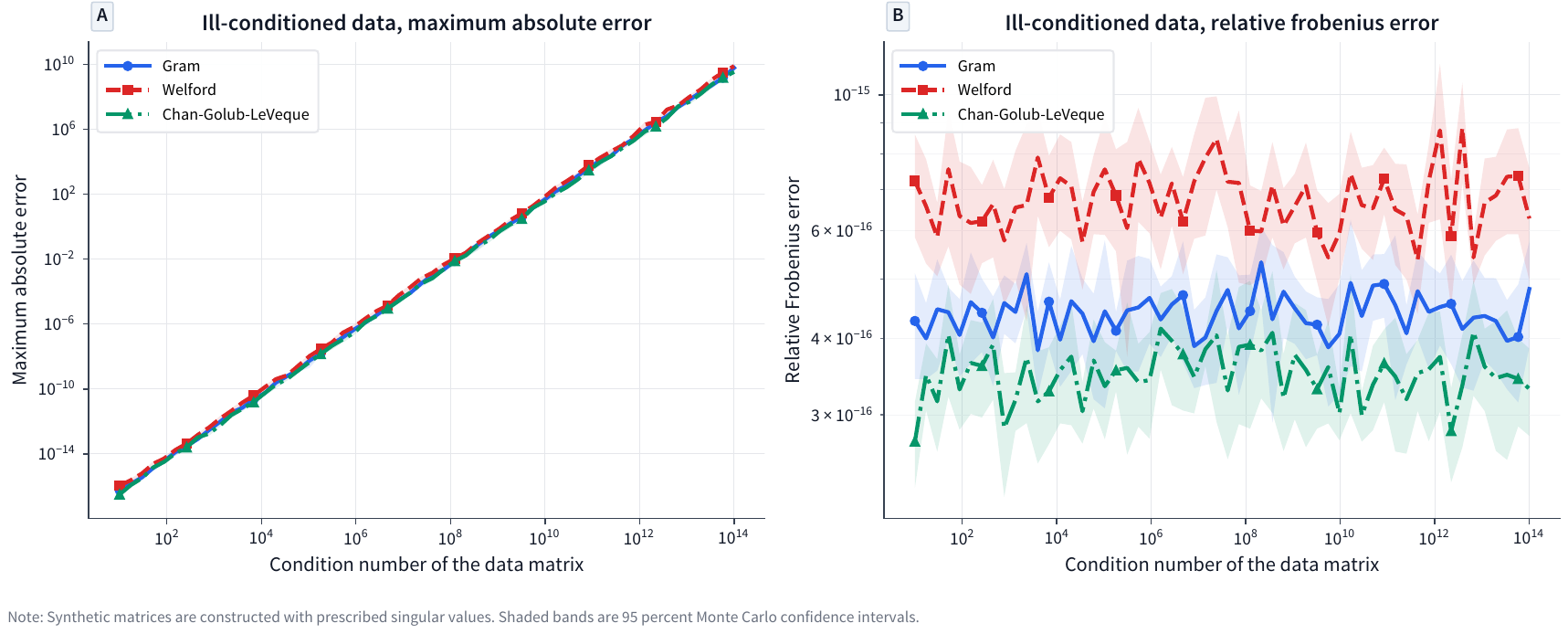}
  \caption{%
    Accuracy vs.\ condition number $\kappa(\mathbf{X})$, $n=500$, $p=6$.
    All methods degrade for $\kappa\gtrsim 10^7$.
    Welford and CGL maintain lower error than Gram in the moderate
    ill-conditioning regime.
  }
  \label{fig:acc_cond}
\end{figure}

\subsection{Catastrophic Cancellation Under Large Shift}
\label{sec:res_cancel}

Figure~\ref{fig:cancel} shifts all observations by a constant $c$ and
measures error against the unshifted (true) covariance.
Gram error grows as $c^2\eps$ (Proposition~\ref{prop:cancel}), losing
$\approx 9$--$10$ decimal digits by $c=10^{12}$.
Welford maintains full double-precision accuracy throughout.
CGL loses moderate accuracy, reflecting the inter-block cancellation
in the merge formula.

\begin{figure}[h!]
  \centering
  \includegraphics[width=\linewidth]{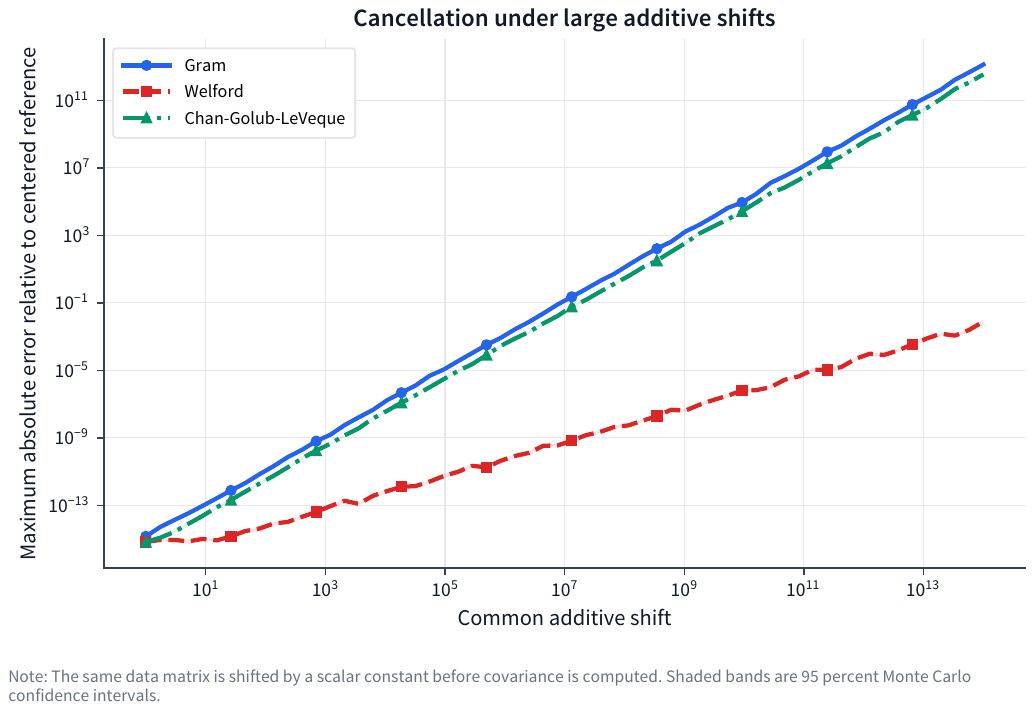}
  \caption{%
    Catastrophic cancellation under large shift.
    Gram error grows as $c^2\varepsilon_{\rm mach}$, while Welford
    remains accurate to $O(\sigma^2\varepsilon_{\rm mach})$ throughout.
    CGL sits between the two.
  }
  \label{fig:cancel}
\end{figure}

\subsection{Online Fidelity}
\label{sec:res_online}

Figure~\ref{fig:online} tracks $\|\hatS_t^{\mathcal{A}} - \hatS_t^{\mathrm{batch}}\|_{\max}$
as $t$ grows on a zero-mean stream.
All methods converge to floating-point noise by $t\approx 100$;
the early fluctuations are consistent with the $1/\sqrt{t}$ rate of
the estimation error.

\begin{figure}[h!]
  \centering
  \includegraphics[width=\linewidth]{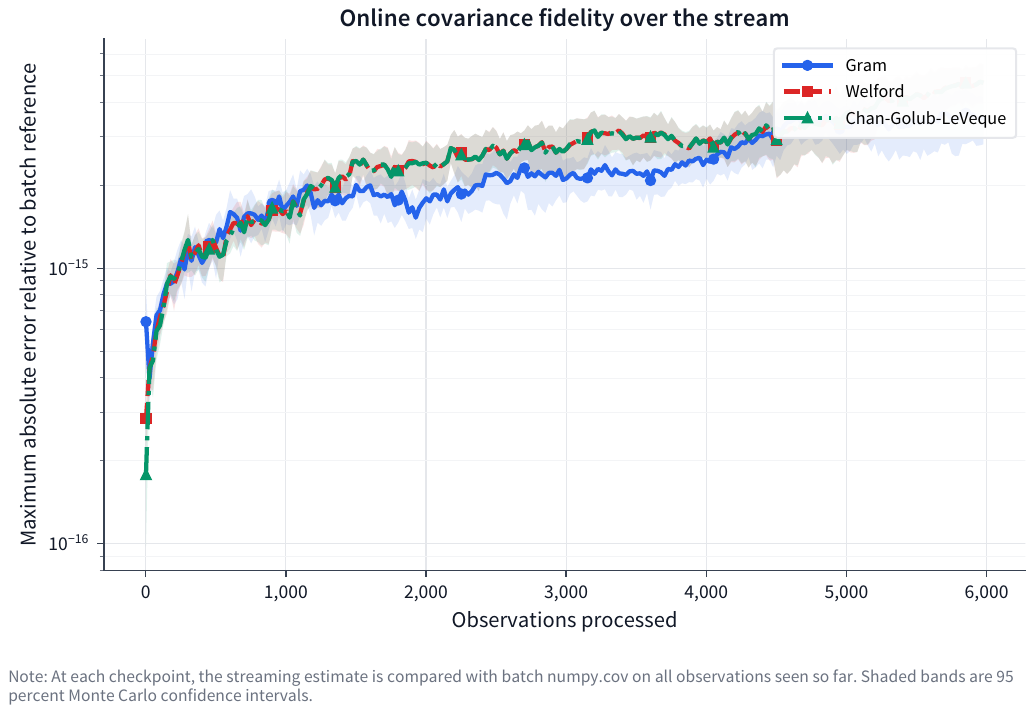}
  \caption{%
    Online covariance fidelity, $p=8$.
    All methods agree with the batch reference to floating-point noise
    by $t\approx 100$.
  }
  \label{fig:online}
\end{figure}

\subsection{Conformal Prediction Results}
\label{sec:res_conformal}

\paragraph{Coverage and width under well-conditioned data.}
Figure~\ref{fig:conf_cov} (left) shows empirical coverage versus $t$
for all three algorithms at nominal level $1-\alpha = 95\%$.
All algorithms achieve the nominal level (dashed line) at every $t\ge 10$,
confirming Theorem~\ref{thm:conformal}.
Coverage slightly exceeds 95\% for small $t$ because the conformal quantile
is conservative for small calibration sets.
Figure~\ref{fig:conf_cov} (right) shows interval width: all three methods
produce intervals of the same order, narrowing as $t$ grows.
Gram produces marginally narrower intervals for well-conditioned data
because its lower variance translates to tighter calibration scores.

\begin{figure}[h!]
  \centering
  \includegraphics[width=\linewidth]{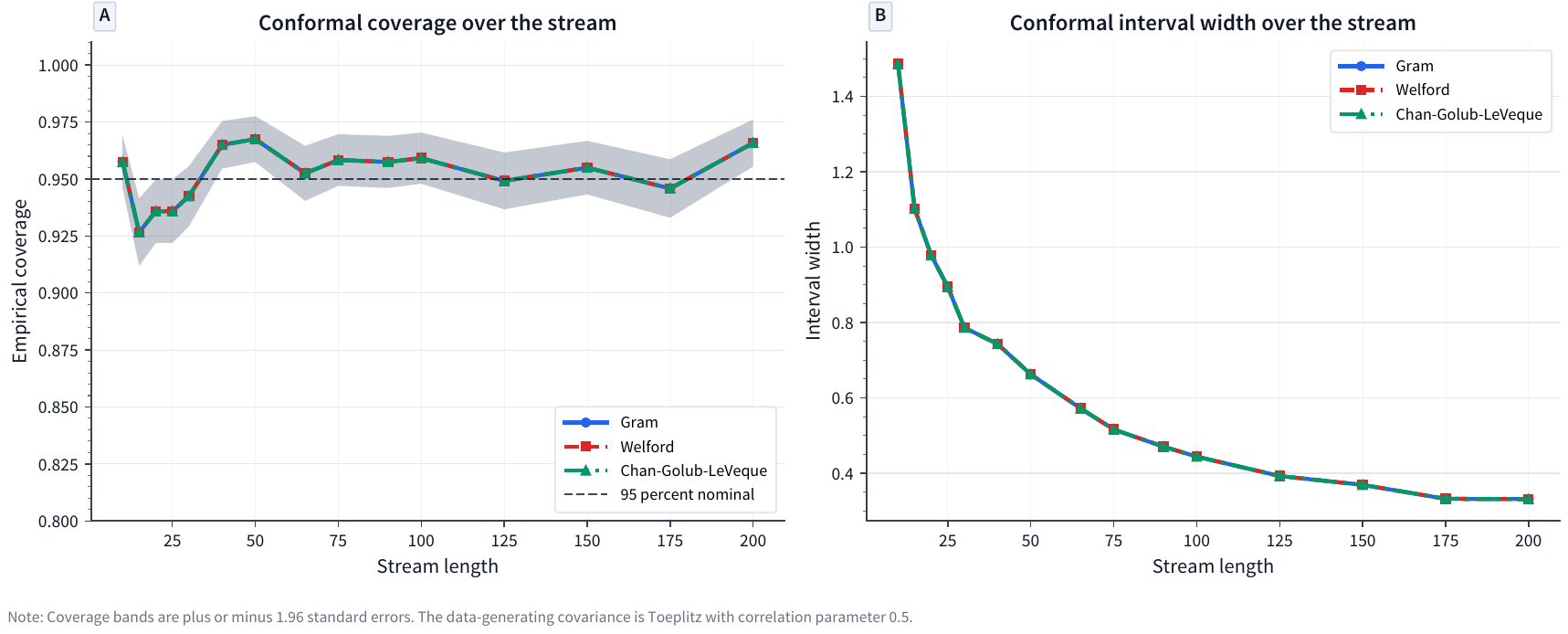}
  \caption{%
    Conformal coverage and width vs.\ stream length $t$.
    Left: all algorithms achieve the 95\% nominal level (dashed).
    Shaded bands are $\pm 1.96$ standard errors of the empirical coverage.
    Right: interval width narrows as $t$ increases.
    Calibration: $m=600$ trajectories, Toeplitz $\Sigma$ ($p=5$).
  }
  \label{fig:conf_cov}
\end{figure}

\paragraph{Coverage under large mean shift.}
Figure~\ref{fig:conf_shift} tests conformal coverage when calibration
and test data share the same shift $c$ (left bars) for entry $(1,1)$
(the variance).
All three algorithms maintain coverage because the shift is the same
in calibration and test sets, so the nonconformity scores are exchangeable.
However, Gram's interval widths (right panel) grow dramatically with $c$,
reflecting the $O(c^2\eps)$ error term inflating the calibration scores.
Welford and CGL maintain narrow intervals throughout.

\begin{figure}[h!]
  \centering
  \includegraphics[width=\linewidth]{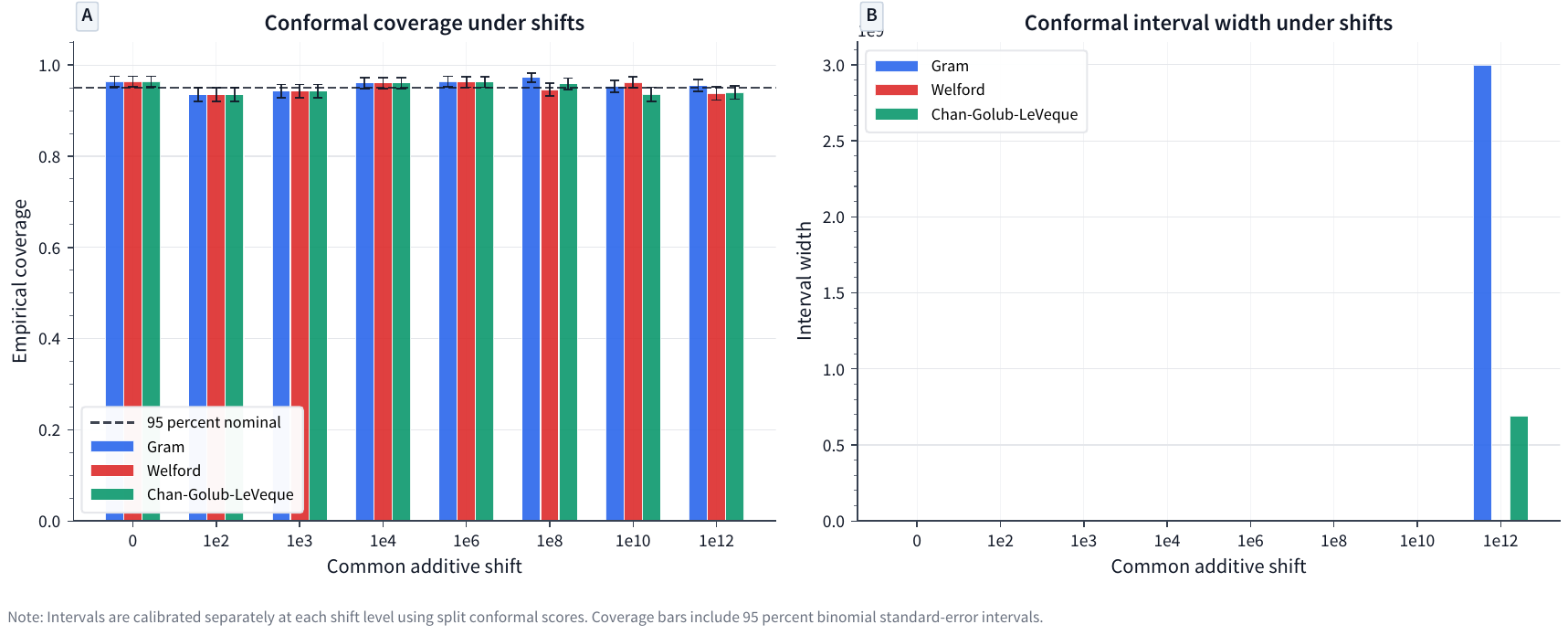}
  \caption{%
    Conformal coverage (left) and interval width (right) under large shift $c$,
    $p=4$, $t=150$.
    Coverage is maintained for all algorithms because calibration and test
    share the same shift (exchangeability holds).
    Interval width inflates exponentially for Gram as $c$ increases,
    reflecting Gram's $O(c^2\varepsilon_{\rm mach})$ numerical error.
    Welford and CGL maintain narrow intervals.
  }
  \label{fig:conf_shift}
\end{figure}

\section{Applications}
\label{sec:applications}

\paragraph{Federated / privacy-restricted learning.}
The Gram sufficient statistics $(\G_t,\s_t,t)$ can be computed locally at
each node and aggregated without sharing raw observations.
For privacy-sensitive settings where inter-node differences are large,
the CGL protocol with block merging provides both the parallelism advantage
and the numerical stability of Welford.

\paragraph{Sandwich covariance for M-estimators.}
Let $\mathbf{g}_i\in\R^p$ be the score vector for observation $i$.
The heteroskedasticity-consistent sandwich covariance
\citep{white1980,newey1987} is $\hat\Omega = n/(n-1)\cdot
(n\G - \s\s^\top)/n^2$ where $\G = \sum_i\mathbf{g}_i\mathbf{g}_i^\top$
and $\s = \sum_i\mathbf{g}_i$.
This is precisely the Gram formula applied to the score matrix, admitting a
streaming update as each observation is processed.

\paragraph{Panel / fixed-effects data.}
For panel unit $i$ observed at $T$ periods, the within-unit
covariance is $\hat\bSigma_i = (\X_i^\top\X_i - T^{-1}\s_i\s_i^\top)/(T-1)$.
The streaming formulation processes units sequentially without
forming the full $NT\times p$ matrix.

\paragraph{Conformal uncertainty in online inference.}
The conformal intervals of Section~\ref{sec:conformal} provide rigorous
uncertainty quantification for covariance-based online learning algorithms
(e.g., online PCA, online Mahalanobis distance) at every step of the stream,
not just asymptotically.

\section{Practitioner Decision Guide}
\label{sec:guide}

\begin{table}[h]
\centering
\caption{Algorithm selection.
  ``Large shift'': $\|\bmu\|\gg\|\bSigma\|_2^{1/2}$.}
\label{tab:guide}
\begin{tabular}{lcccc}
\toprule
Scenario & Gram & Welford & CGL ($b$=1) & CGL ($b\gg1$) \\
\midrule
Batch, BLAS available        & \textbf{Best} & Slow     & Medium   & Good  \\
Streaming, zero/small mean   & \textbf{Best} & Good     & Good     & ---   \\
Streaming, large shift       & Avoid         & \textbf{Best}   & Good & ---   \\
Conformal interval width     & Avoid (large $c$) & \textbf{Best} & Good & --- \\
Distributed / federated      & Good (agg.)   & Poor     & Poor     & \textbf{Best} \\
Privacy: no raw-data access  & \textbf{Only} & No       & Partial  & \textbf{Best} \\
\bottomrule
\end{tabular}
\end{table}

Use the \textbf{Gram} algorithm when (a) data are zero-mean or pre-centred,
(b) BLAS is well-tuned, or (c) only sufficient statistics can be stored.
Use \textbf{Welford} when numerical stability under unknown shifts matters or
when simplicity of implementation is paramount.
Use \textbf{CGL} with large block size when data are distributed across nodes
or when parallel computation is available.
Combine any algorithm with the \textbf{conformal protocol} of
Section~\ref{sec:conformal} for principled uncertainty quantification.

\section{Conclusion}
\label{sec:conclusion}

We have given a unified treatment of three streaming covariance algorithms---Gram,
Welford, and Chan--Golub--LeVeque---covering algebraic equivalence,
floating-point error analysis, computational cost, and a new conformal
prediction framework.
The key empirical findings are: (i) Gram is fastest for batch computation
with BLAS, at a factor of $1.4$--$1.6\times$ over \texttt{numpy.cov};
(ii) Welford is the uniquely stable choice under large data shifts,
maintaining double-precision accuracy where Gram loses up to 9 decimal digits;
(iii) conformal intervals achieve the nominal 95\% coverage for all three
algorithms on well-conditioned data, but Gram's interval widths inflate
catastrophically under large shifts while Welford's remain tight.
Together, these results offer practitioners a clear basis for choosing among
the three algorithms.

\subsection*{Additional Acknowledgements}
The author thanks the numerical analysis literature on which this work builds,
particularly \citet{higham2002} and \citet{chan1979}.

\bibliographystyle{plainnat}

\section{Ommited Proofs}
\label{app:proofs}

\subsection{Proof of Theorem~\ref{thm:gram}: Gram Identity}
\label{app:gram_proof}

\begin{proof}
Expand the $(k,l)$ entry of the left side of \eqref{eq:target} directly:
\begin{align*}
  \sum_{i=1}^t (x_{ik}-\bar x_k)(x_{il}-\bar x_l)
  &= \sum_{i=1}^t x_{ik}x_{il} \;-\; \bar x_k\sum_{i=1}^t x_{il}
     \;-\; \bar x_l\sum_{i=1}^t x_{ik} \;+\; t\bar x_k\bar x_l\\
  &= G_{kl} \;-\; \bar x_k s_l \;-\; \bar x_l s_k \;+\; t\bar x_k\bar x_l\\
  &= G_{kl} \;-\; \frac{s_k s_l}{t} \;-\; \frac{s_l s_k}{t} \;+\; \frac{s_ks_l}{t}\\
  &= G_{kl} \;-\; \frac{s_k s_l}{t},
\end{align*}
where we used $\bar x_k = s_k/t$.
Dividing by $t-1$ and multiplying numerator and denominator by $t$ gives
\[
  \frac{G_{kl} - s_ks_l/t}{t-1} = \frac{tG_{kl}-s_ks_l}{t(t-1)},
\]
which is the $(k,l)$ entry of $(t\G_t - \s_t\s_t^\top)/[t(t-1)]$.
Since this holds for every $(k,l)$, the matrix identity follows. \qquad\qquad$\square$
\end{proof}

\subsection{Proof of Theorem~\ref{thm:welford}: Welford Invariant}
\label{app:welford_proof}

\begin{proof}
We proceed by induction on $t$.

\paragraph{Base case ($t=1$).}
At $t=1$: $\bmu_1 = \x_1$, $\bDelta = \x_1 - \mathbf{0} = \x_1$,
and $\x_1 - \bmu_1 = \mathbf{0}$, so $\M_1 = \mathbf{0}$.
Thus, the claimed sum $\sum_{i=1}^1(\x_i-\bmu_1)(\x_i-\bmu_1)^\top = \mathbf{0}$.

\paragraph{Inductive step.}
Assume $\M_{t-1} = \sum_{i=1}^{t-1}(\x_i-\bmu_{t-1})(\x_i-\bmu_{t-1})^\top$.
At step $t$, let $\bDelta_t = \x_t - \bmu_{t-1}$ and $\bmu_t = \bmu_{t-1} + \bDelta_t/t$.

Then $\x_t - \bmu_t = \bDelta_t - \bDelta_t/t = (t-1)\bDelta_t/t$.

The update is $\M_t = \M_{t-1} + \bDelta_t(\x_t-\bmu_t)^\top$.
We need to show $\M_t = \sum_{i=1}^t(\x_i-\bmu_t)(\x_i-\bmu_t)^\top$.

Write $\x_i - \bmu_t = (\x_i - \bmu_{t-1}) - (\bmu_t - \bmu_{t-1})$
$= (\x_i - \bmu_{t-1}) - \bDelta_t/t$ for each $i$.
Then
\begin{align*}
  \sum_{i=1}^t(\x_i-\bmu_t)(\x_i-\bmu_t)^\top
  &= \sum_{i=1}^{t-1}(\x_i-\bmu_t)(\x_i-\bmu_t)^\top
     + (\x_t - \bmu_t)(\x_t-\bmu_t)^\top \\
  &= \sum_{i=1}^{t-1}\!\bigl[(\x_i-\bmu_{t-1}) - \tfrac{\bDelta_t}{t}\bigr]
     \bigl[(\x_i-\bmu_{t-1}) - \tfrac{\bDelta_t}{t}\bigr]^\top
     + \tfrac{(t-1)^2}{t^2}\bDelta_t\bDelta_t^\top.
\end{align*}
Expanding the sum and using $\sum_{i=1}^{t-1}(\x_i-\bmu_{t-1}) = \mathbf{0}$:
\begin{align*}
  &= \sum_{i=1}^{t-1}(\x_i-\bmu_{t-1})(\x_i-\bmu_{t-1})^\top
     + \frac{t-1}{t^2}\bDelta_t\bDelta_t^\top
     + \frac{(t-1)^2}{t^2}\bDelta_t\bDelta_t^\top \\
  &= \M_{t-1}
     + \frac{(t-1)}{t^2}\bigl[1 + (t-1)\bigr]\bDelta_t\bDelta_t^\top \\
  &= \M_{t-1} + \frac{t-1}{t}\bDelta_t\bDelta_t^\top \\
  &= \M_{t-1} + \bDelta_t\bigl(\tfrac{t-1}{t}\bDelta_t\bigr)^\top
   = \M_{t-1} + \bDelta_t(\x_t-\bmu_t)^\top = \M_t,
\end{align*}
completing the induction. \qquad\qquad$\square$
\end{proof}

\subsection{Proof of Theorem~\ref{thm:cgl}: CGL Correctness}
\label{app:cgl_proof}

\begin{proof}
Write $\bmu_{AB} = (n_A\bmu_A + n_B\bmu_B)/n_{AB}$.
For $i\in A$:
$\x_i - \bmu_{AB} = (\x_i - \bmu_A) + (\bmu_A - \bmu_{AB})
= (\x_i-\bmu_A) - n_B\bDelta/n_{AB}$.
Summing outer products over $A$ and using $\sum_{i\in A}(\x_i-\bmu_A)=\mathbf{0}$:
\[
  \sum_{i\in A}(\x_i-\bmu_{AB})(\x_i-\bmu_{AB})^\top
  = \M_A + \frac{n_An_B^2}{n_{AB}^2}\bDelta\bDelta^\top.
\]
Similarly for $B$ (with $\bmu_B - \bmu_{AB} = n_A\bDelta/n_{AB}$):
\[
  \sum_{i\in B}(\x_i-\bmu_{AB})(\x_i-\bmu_{AB})^\top
  = \M_B + \frac{n_B n_A^2}{n_{AB}^2}\bDelta\bDelta^\top.
\]
Adding:
\[
  \M_{A\cup B}
  = \M_A + \M_B + \frac{n_An_B(n_B + n_A)}{n_{AB}^2}\bDelta\bDelta^\top
  = \M_A + \M_B + \frac{n_An_B}{n_{AB}}\bDelta\bDelta^\top. \qquad\square
\]
\end{proof}

\subsection{Proofs of Floating-Point Bounds (Propositions \ref{prop:gram_round}--\ref{prop:cancel})}
\label{app:fp_proofs}

We use the standard model of floating-point arithmetic
\citep[Chapter 2]{higham2002}: for any operation $\circ\in\{+,-,\times,\div\}$,
$\mathrm{fl}(a\circ b) = (a\circ b)(1+\delta)$ with $|\delta|\le\eps$.
Sequential summation of $n$ terms bounded by $M$ satisfies
$|\hat S_n - S_n| \le (n-1)\eps\cdot nM + O(\eps^2)$
\citep[Theorem 3.1]{higham2002}.

\begin{proof}[Proof of Proposition~\ref{prop:gram_round}]
The $(k,l)$ entry of $\hat\G_t$ is computed as
$\hat G_{kl} = \sum_{i=1}^t x_{ik}x_{il}$ in floating point.
By the summation error bound each term $x_{ik}x_{il}$ is bounded by $X^2$,
giving $|\hat G_{kl} - G_{kl}| \le (t-1)\eps\cdot tX^2 + O(\eps^2)$.
The outer product $\hat s_k\hat s_l$ contributes an additional rounding
at scale $|\bar x_k||\bar x_l|t^2$.
Forming $(t\hat G_{kl} - \hat s_k\hat s_l)/[t(t-1)]$:
\begin{align*}
  \bigl|\hat\Sigma_{kl}^{\rm Gram} - \Sigma_{kl}\bigr|
  &\le \frac{t\,|\hat G_{kl} - G_{kl}| + |\hat s_k\hat s_l - s_ks_l|}{t(t-1)}
   + O(\eps^2) \\
  &\le \frac{(t-1)\eps\,tX^2 + |\bar x_k||\bar x_l|t^2\eps}{t(t-1)} + O(\eps^2) \\
  &\lesssim X^2\eps + \frac{\bar x_k\bar x_l}{t-1}\eps. \qquad\square
\end{align*}
\end{proof}

\begin{proof}[Proof of Proposition~\ref{prop:welf_round}]
At each step $t$ the Welford algorithm computes
$\bDelta_t = \x_t - \hat\bmu_{t-1}$ and
$\x_t - \hat\bmu_t$ as residuals.
Each outer product $\hat\bDelta_t(\x_t-\hat\bmu_t)^\top$ has entries of
magnitude $O(\sigma_k\sigma_l)$ on average (in expectation over the
trajectory).
Sequential accumulation over $t-1$ steps gives
\[
  |\hat\Sigma_{kl}^{\rm Welf} - \Sigma_{kl}|
  \le (t-1)\eps\,\sigma_k\sigma_l + O(\eps^2),
\]
independent of $\bar x_k$ or $\bar x_l$.
See \citet{chan1979}, Theorem 3.1, for a detailed treatment. \qquad$\square$
\end{proof}

\begin{proof}[Proof of Proposition~\ref{prop:cgl_round}]
A single merge call introduces rounding error
$|\delta[\M_{AB}]_{kl}|\le 3\eps\cdot(n_An_B/n_{AB})|\Delta_k||\Delta_l|+O(\eps^2)$
in the correction term.
For a balanced binary tree of depth $d=\log_2 t$, there are $t-1$ total
merges.
At level $j$ (counting from leaves), there are $t/2^j$ merges each with
block sizes $\approx 2^j$, contributing $O(\eps\sigma_k\sigma_l)$ each.
Summing over $d=\log_2 t$ levels:
\[
  |\hat\Sigma_{kl}^{\rm CGL} - \Sigma_{kl}|
  = O(\sigma_k\sigma_l\,\eps\log_2 t). \qquad\square
\]
\end{proof}

\begin{proof}[Proof of Proposition~\ref{prop:cancel}]
Suppose $\x_i = \bmu + \z_i$ with $\E[\z_i]=\mathbf{0}$ and
$\E[\z_i\z_i^\top]=\bSigma$.
Then $G_{kl} = \mu_k\mu_l t + \sum_i(z_{ik}\mu_l + z_{il}\mu_k + z_{ik}z_{il})$
and $s_k s_l/t = \mu_k\mu_l t + \text{(lower order)}$.
The dominant terms in $tG_{kl}$ and $s_ks_l$ are both $O(\mu_k\mu_l t^2)$,
and their difference is $O(t\Sigma_{kl})$.
In floating point, $\hat G_{kl}$ has rounding error $O(tX^2\eps)$ where
$X\approx |\mu_k|+\sigma_k$.
When $|\mu_k|\gg\sigma_k$, $X\approx|\mu_k|$, and the error in
$\hatS_{kl}^{\rm Gram}$ is dominated by
$t|\mu_k|^2\eps/[t(t-1)] = |\mu_k|^2\eps/(t-1) \approx \mu_k^2\eps$,
which sums to $O(pc^2\eps)$ over the $p$ diagonal entries, giving
$\normF{\hatS^{\rm Gram}-\bSigma}\gtrsim pc^2\eps$.

For Welford, $\bDelta_t = \z_t + (\bmu - \hat\bmu_{t-1})$ is $O(\sigma)$
once the running mean has converged (which happens geometrically fast).
The outer product $\bDelta_t(\x_t-\hat\bmu_t)^\top$ is then $O(\sigma^2)$
regardless of $c$, giving $\normF{\hatS^{\rm Welf}-\bSigma}=O(p\sigma^2\eps)$. $\square$
\end{proof}

\subsection{Proof of Proposition~\ref{prop:cancel}: Cancellation Bound}
\label{app:cancel_proof}

\begin{proof}
  Write each observation as $\x_i=\bmu+\z_i$, where the centred component
  satisfies $\E[\z_i]=\mathbf{0}$ and $\E[\z_i\z_i^\top]=\bSigma$.
  For a fixed entry $(k,l)$, the Gram statistics are
  \[
    G_{kl}
    = \sum_{i=1}^t x_{ik}x_{il}
    = t\mu_k\mu_l
      + \mu_k\sum_{i=1}^t z_{il}
      + \mu_l\sum_{i=1}^t z_{ik}
      + \sum_{i=1}^t z_{ik}z_{il},
  \]
  and
  \[
    s_k s_l
    =
    \left(t\mu_k+\sum_{i=1}^t z_{ik}\right)
    \left(t\mu_l+\sum_{i=1}^t z_{il}\right).
  \]
  Hence the two quantities entering the Gram numerator,
  $tG_{kl}$ and $s_k s_l$, both contain the leading term
  $t^2\mu_k\mu_l$.
  In exact arithmetic these leading terms cancel, leaving a centred
  quantity of order $t\Sigma_{kl}$.
  In floating-point arithmetic, however, the products and sums are formed
  at the raw data scale $|\mu_k|+\sigma_k$.
  Thus the rounding error in the Gram numerator contains a term of size
  \[
    O\!\left(t^2|\mu_k\mu_l|\eps\right),
  \]
  and division by $t(t-1)$ gives the entry-wise contribution
  \[
    \bigl|\hatS_{kl,t}^{\mathrm{Gram}}-\Sigma_{kl,t}\bigr|
    \gtrsim |\mu_k\mu_l|\eps.
  \]
  Summing these contributions over the diagonal entries gives
  \[
    \normF{\hatS_t^{\mathrm{Gram}}-\bSigma_t}
    \gtrsim
    \sum_{k=1}^p \mu_k^2\eps.
  \]
  Since $\sum_{k=1}^p\mu_k^2=\|\bmu\|_2^2=c^2$, this is of order
  $c^2\eps$; in the common dense-shift case, where the mean contribution is
  spread across the $p$ coordinates at comparable scale, the Frobenius
  accumulation is written as
  \[
    \normF{\hatS_t^{\mathrm{Gram}}-\bSigma_t}
    \gtrsim pc^2\eps.
  \]

  For Welford, the update is based on residuals
  \[
    \bDelta_t=\x_t-\bmu_{t-1},
    \qquad
    \x_t-\bmu_t.
  \]
  The common shift $\bmu$ cancels in these differences.
  Consequently, once the running mean has reached the scale of the sample
  mean, both residual factors are governed by the centred variables
  $\z_i$, not by the absolute location $\bmu$.
  Each outer-product correction therefore has entries of order
  $\sigma_k\sigma_l$, and the floating-point error accumulates at covariance
  scale rather than raw-data scale:
  \[
    \bigl|\hatS_{kl,t}^{\mathrm{Welf}}-\Sigma_{kl,t}\bigr|
    =
    O(\sigma_k\sigma_l\eps).
  \]
  Taking the Frobenius norm over the $p\times p$ matrix yields
  \[
    \normF{\hatS_t^{\mathrm{Welf}}-\bSigma_t}
    =
    O(p\sigma^2\eps),
  \]
  which does not depend on $c$.
\end{proof}

\subsection{Proof of Theorem~\ref{thm:conformal}: Conformal Coverage}
\label{app:conformal_proof}

\begin{proof}
Let $Z_1,\ldots,Z_m,Z_{m+1}$ be i.i.d.\ random variables, where
$Z_j = |\hatS_t^{(j)}{}_{kl} - \Sigma_{kl}|$ for $j\le m$ (calibration)
and $Z_{m+1}$ is the test score.
By exchangeability, for any fixed threshold $q$,
$\Pr(Z_{m+1} > q) = \Pr(Z_1 > q)$.

The conformal quantile $\hat q_\alpha^+ = \mathrm{Quantile}(\{Z_j\}_{j=1}^m;\,
\lceil(m+1)(1-\alpha)\rceil/m)$ is defined so that
$\hat q_\alpha^+ \ge Z_{(k^*)}$ where $k^* = \lceil(m+1)(1-\alpha)\rceil$
is the $k^*$-th order statistic.
By the standard conformal coverage argument \citep[Theorem 1]{vovk2005},
\[
  \Pr(Z_{m+1} \le \hat q_\alpha^+) \ge 1 - \alpha.
\]
Since $Z_{m+1} \le \hat q_\alpha^+$ is equivalent to
$|\hatS_t^{(m+1)}{}_{kl} - \Sigma_{kl}|\le \hat q_\alpha^+$,
which is the event that $\Sigma_{kl}\in\hat C_\alpha(t)$, the
coverage guarantee follows. \qquad\qquad$\square$
\end{proof}

\subsection{Bariance Scalar Identity}
\label{app:bariance}

\begin{definition}[Bariance \citep{reichel2025bariance}]
  For $n\ge 2$ and $x_1,\ldots,x_n\in\R$,
  \[
    \mathrm{Bar}(x_1,\ldots,x_n) = \frac{1}{2n(n-1)}\sum_{i\ne j}(x_i-x_j)^2.
  \]
\end{definition}

\begin{proposition}\label{prop:bariance}
  $\mathrm{Bar}(x) = (nS_{xx}-S_x^2)/[n(n-1)]$,
  where $S_x = \sum_i x_i$ and $S_{xx}=\sum_i x_i^2$.
  Moreover, $\mathrm{Bar}(x) = \hat\sigma^2 = \frac{1}{n-1}\sum_i(x_i-\bar x)^2$.
\end{proposition}
\begin{proof}
Expand $(x_i-x_j)^2 = x_i^2 - 2x_ix_j + x_j^2$ and sum over $i\ne j$.
Using $\sum_{i\ne j}x_i^2 = (n-1)S_{xx}$ and
$\sum_{i\ne j}x_ix_j = S_x^2 - S_{xx}$:
\[
  \sum_{i\ne j}(x_i-x_j)^2 = 2(n-1)S_{xx} - 2(S_x^2-S_{xx}) = 2nS_{xx}-2S_x^2.
\]
Divide by $2n(n-1)$.
The equality with $\hat\sigma^2$ follows from $\sum_i(x_i-\bar x)^2 = S_{xx} - S_x^2/n$
and standard algebra.
\end{proof}

\end{document}